\documentclass[12pt]{article}

\usepackage[colorlinks=false]{hyperref}
\usepackage[dvips]{graphicx}
\usepackage[usenames,dvipsnames]{color}
\usepackage{epsfig}
\usepackage{rotating}
\usepackage{amssymb}
\usepackage{amsmath,amsfonts}
\usepackage{multirow}
\usepackage[latin1]{inputenc}
\usepackage{verbatim}
\usepackage[tableposition=top,font=small,labelfont=bf,format=hang]{caption}
\usepackage{booktabs}

\newtheorem{lem}{Lemma}
\newtheorem{thm}{Theorem}
\newtheorem{rem}{Remark}
\newtheorem{examp}{Example}
\newtheorem{prop}{Proposition}

\newtheorem{coro}{Corollary}

\newcommand{\Mu}{\mathcal{M}}

\newcommand{\F}{\mathbb{F}}

\newenvironment{proof}{\noindent\textbf{Proof.}\quad}
{\hspace{\stretch{1}}%
\rule{1ex}{1ex}\\}

\begin{document}
 \author{ Minjia Shi\thanks{smjwcl.good@163.com}, Shukai Wang\thanks{wangshukai\_2017@163.com},  Jon-Lark Kim\thanks{jlkim@sogang.ac.kr},
Patrick Sol\'e\thanks{{sole@enst.fr}}
\thanks{Minjia Shi and Shukai Wang are with School of Mathematical Sciences, Anhui University, Hefei, 230601, China. Jon-Lark Kim is with Sogang University, Seoul, South Korea. Patrick Sol\'e is with Aix Marseille Univ, CNRS, Centrale Marseille, I2M, Marseille, France. }}
\title{Self-orthogonal codes over a non-unital ring\\ and \\ combinatorial matrices\thanks{This research is supported by National Natural Science Foundation of China (12071001, 61672036), Excellent Youth Foundation of Natural Science Foundation of Anhui Province (1808085J20), the Academic Fund for Outstanding Talents in Universities (gxbjZD03).}}
\date{}
\maketitle
\begin{abstract}
There is a  local ring $E$ of order $4,$ without identity for the multiplication, defined by generators and relations as
 $E=\langle a,b \mid 2a=2b=0,\, a^2=a,\, b^2=b,\,ab=a,\, ba=b\rangle.$
 We study a special construction of self-orthogonal codes over $E,$ based on combinatorial matrices related to two-class association schemes,
  Strongly Regular Graphs (SRG), and Doubly Regular Tournaments (DRT).
  We construct quasi self-dual codes over $E,$ and Type IV codes,
 that is, quasi self-dual codes whose all codewords have even Hamming weight.
 All these codes can be represented as formally self-dual additive codes over $\F_4.$ The classical invariant theory bound for the weight enumerators
 of this class of codes
 improves the known bound on the minimum distance of Type IV codes over $E.$
\end{abstract}
{\bf Keywords:} {rings, codes, formally self-dual codes, Type IV codes. }\\
{\bf MSC(2010):} Primary 94 B05, Secondary 16 A10.
\section{Introduction}

Since the the celebrated theorem of Gleason and Prange \cite{GP}, formally self-dual codes over $\F_4$ with even weights, also known as Type IV codes have been studied extensively \cite[Chap. 19]{MS},
\cite{NRS}.
In \cite{D+} this notion was extended over the three rings of order four that are not a field, namely
 $\mathbb Z_4$, $\mathbb F_2 + u \mathbb F_2$, and $\mathbb F_2 + v \mathbb F_2$.
Recently, a further extension was accomplished over a non commutative non-unital ring  in~\cite{E}. The concept of self-dual code is replaced there by quasi self-dual (QSD) code
that is self-orthogonal of length $n,$ with $2^n$ codewords. Type IV codes are then defined as QSD codes, whose Hamming weights of all codewords are even.
With every linear $E$-code is attached an additive $\F_4$-code obtained by forgetting the ring structure; this allows to use the additive codes package of \cite{DS} for numerical computations.
Kim and Ohk~\cite{KimOhk} showed that quasi self-dual codes over that  ring $E$ can be applied to DNA codes in the sense that the GC-content concept can be described
by a multiple of an element in the ring. They also improved the classification of QSD codes over $E$ up to lengths $8.$ The Lee weight defined below is based on this DNA application.

In this paper, we study a special construction of QSD codes over $E,$ based on combinatorial matrices related to two-class association schemes,
  Strongly Regular Graphs (SRG), and Doubly Regular Tournaments (DRT). This is a generalization from fields to rings of the approach of \cite{T}.
  We construct QSD codes and Type IV codes over $E.$ Along the way, we improve the upper bound on the minimum distance of Type IV codes from \cite{E} by a multiplicative factor,
  by an application of the classical invariant bound for the minimum distance of extremal Type IV codes over $\F_4.$
  Some numerical results validate our approach.

The material is arranged in the following way. Section 2 collects the notions and notations required for the rest of the paper. Section 3 studies our special construction.
Section 4 develops the needed theory of combinatorial matrices from designs, SRGs and DRTs. Section 5 concludes the article.

\section{Background}
\subsection{Binary codes}
Denote by $wt(x)$ the Hamming weight of $x \in \F_2^n.$
The dual of a binary linear code $C$ is denoted by $C^\bot$ and defined as
$$C^{\bot}=\{ y \in \F_2^n \mid \forall x \in C, \,( x,y)=0\} ,$$
where $(x,y)=\sum_{i=1}^nx_iy_i,$ denotes the standard inner product.
A code $C$ is {\bf self-orthogonal} if it is included in its dual: $C \subseteq C^{\bot}.$
Two binary codes are {\bf equivalent} if there is a permutation of coordinates that maps one to the other.
\subsection{Quaternary codes}
An {\bf additive code} of length $n$ over $\F_4$ is an additive subgroup of $\F_4^n$. It is a free $\F_2$ module with $4^k$ elements for some $k\le n$
(here $2k$ is an integer, but $k$ may be half-integral).  Using a {\bf generator matrix} $G$, such a code can be represented as the $\F_2$-span of its rows.
With every linear $E$ code $C$ is attached an {\bf additive} $\F_4$ code $\phi(C)$  by the substitution
$$0 \to 0,a \to \omega, b\to \omega^2 , c\to 1, $$ where $\F_4=\F_2[\omega].$
Note that the reverse substitution attaches to every additive $\F_4$ code an additive subgroup of $E^n,$ which may or may not be linear.

Besides the Hamming weight of a vector, we might consider its Lee weight as follows:
 $$wt_L(0)=0, wt_L(a)=wt_L(b)=1, wt_L(c)=2.$$

\subsection{Ring theory}
Consider the ring of order $4$ defined by two generators $a$ and $b$ by the relations
$$E=\langle a,b \mid 2a=2b=0,\, a^2=a,\, b^2=b,\,ab=a,\, ba=b\rangle.$$
The ring $E$ is a non unital, non-commutative ring of order $4,$ of characteristic two\cite{E, F}.
Thus, $E$  consists of four elements $E=\{0,a,b,c\},$ with $c=a+b.$
Its multiplication table is as follows.
$$\begin{array}{|c|c|c|c|c|}
  \hline
 \times & 0 & a&  b& c\\\hline
  0 &  0&  0&  0& 0\\\hline
  a &  0& a & a & 0 \\ \hline
  b&  0& b & b& 0 \\ \hline
  c & 0& c & c &0\\
  \hline
\end{array}$$
From this table, we deduce that
this ring is not commutative, and non-unital. It is local with maximal ideal $J=\{0,c\},$
and residue field $E/J=\F_2=\{0,1\},$ the finite field of order $2.$

Denote by $\alpha : E \rightarrow  E/J=\F_2$, the map of reduction modulo $J$. Thus
$\alpha(0) = \alpha(c) = 0$, and $\alpha(a) = \alpha(b) = 1$. This map is extended in the natural way in a map from $E^n$ to $\mathbb F_2^n$.

\subsection{Codes over $E$}

A {\bf linear} $E$-code of length $n$ is a one-sided $E$-submodule of $E^n.$
Let $C$ be a code of length $n$ over $E.$ With that code we associate two binary codes of length $n:$
\begin{enumerate}
 \item [(1)] the {\bf residue code} defined by
$res(C)=\{\alpha(y) \mid y \in C\},$
 \item [(2)]  the {\bf torsion code} defined by
$tor(C)=\{x \in \F_2^n \mid  cx \in C\}.$
\end{enumerate}

The {\bf right dual} $C^{\bot_R}$ of $C$ is the right module defined by

$$C^{\bot_R}=\{ y \in E^n \mid \forall x \in C, \,( x,y)=0\} . $$

The {\bf left dual} $C^{\bot_L}$ of $C$ is the left module defined by

$$C^{\bot_L}=\{ y \in E^n \mid \forall x \in C, \, ( y,x)=0\} . $$

An $E$-code $C$ is  {\bf self-orthogonal} if
 $$ \forall x,y \in C, (x,y)=0.$$
 Clearly, $C$ is  {\bf self-orthogonal} if and only if $C\subseteq C^{\bot_L}.$ Likewise, $C$ is  {\bf self-orthogonal} if and only if $C\subseteq C^{\bot_R}.$
 Thus, for a self-orthogonal code $C,$ we always have $C\subseteq C^{\bot_L} \cap C^{\bot_R}.$
 An $E$-code of length $n$ is  {\bf Quasi Self-Dual} (QSD for short ) if it is self-orthogonal and of size $2^n.$
A QSD code is {\bf Type IV} if all its codewords have even weight.

The following result went unnoticed in \cite{E}, and improves on the previously known upper bound $d
\le 2 \lfloor \frac{n+2}{4} \rfloor$ for the minimum Hamming distance $d$ of a Type IV $E$-code of length $n$.

{\thm If $C$ is a Type IV $E$-code of length $n$, then it is formally self-dual for the Hamming weight enumerator, and its minimum distance is $\le 2(\lfloor \frac{n}{6}\rfloor+1).$}

\begin{proof} The first statement follows by specialization of variables in the MacWilliams relation for the joint weight enumerator of the residue and torsion code \cite[Prop. 2]{E}.
 The second statement follows by the standard argument used to prove the same bound for Type IV codes over $\F_4$ \cite[Chap. 19, (69)]{MS}. Note that the Hamming weight enumerator
 of a Type IV code over $E$
 belongs to the same ring of
 invariants as that of a Type IV code over $\F_4.$
\end{proof}

Theorem 1 gives a construction of additive formally self-dual even codes over $\F_4.$
{\coro If $C$ is Type IV then $\phi(C)$ is an additive formally self-dual even code.}

\begin{proof}
 The results follow by the fact that $C$ and $\phi(C)$ have the same Hamming weight enumerator.
\end{proof}

We now study the residue and torsion code of a QSD code over $E.$

\begin{thm}[\cite{E}]
For any QSD $E$-linear code $\mathcal{C}$, we have
\begin{enumerate}
\renewcommand{\labelenumi}{$(\theenumi)$}
\item $res(\mathcal{C}) \subseteq {res(\mathcal{C})}^{\perp}$,

\item $tor(\mathcal{C}) = {res(\mathcal{C})}^{\perp}$,

\item $dim(\mathcal{C}) = dim(res(\mathcal{C})) + dim(tor(\mathcal{C}))$.
\end{enumerate}
\end{thm}

We can characterize QSD codes over $E$ amongst linear codes over $E$ as a function of their residue code in the following theorem.

\begin{thm}[\cite{E}]\label{E}
Let $B$ be a self-orthogonal binary $\left[n,k_1\right]$ code, where $0 \le k_1 \le n/2$. The code $\mathcal{C}$ over the ring $E$ defined by the relation
\begin{equation*}
\mathcal{C}=aB+cB^{\perp}
\end{equation*}
is a QSD code. Its residue code is $B$ and its torsion code is $B^{\perp}$. Conversely, any QSD code $C$ can be built in that way by taking for $B$ the residue code of $C.$
\end{thm}

By Theorem \ref{E}, we know that the classification of QSD $E$-codes is equivalent to the classification of their residue codes. Moreover, the following result is straightforward, but useful. The easy proof is ommitted.

\begin{thm}\label{thm-min-dinstance}
The minimum distance $d(C)$ of a QSD code $C$ defined by $\mathcal{C}=aB+cB^{\perp}$, where $B$ is a self-orthogonal binary code, is less than or equal to $\min \{d(B), d(B^{\perp}) \}$. If $B$ is a self-dual binary code, then $d(C)=d(B)$.

\end{thm}

\section{Construction}\label{sec:3}

Consider the code $C(M)$ of length $2n$ with a generator matrix of the form
$$G=(xI,yM)$$
where $x, y \in E,$ $I$ is the identity matrix, and $M$ is a binary matrix satisfying

$$MM^T=\lambda I +\mu J+\nu M,$$
where $\lambda, \mu, \nu \in \F_2,$ and $J$ is the all-one matrix.

\begin{table}[!hbpt]
\centering
\caption{Conditions of self-orthogonal codes}\label{tab:11}
\begin{tabular}{|c|c|c|c|}
\hline
$n$ & $\lambda$ & $\mu$ & $\nu$ \\ \hline
 any & 1 & 0 & 0  \\ \hline
 odd & 1 & 1 & 0  \\ \hline
 any & 0 & 0 & 1  \\ \hline
 even & 0 & 1 & 1  \\ \hline
 odd & 0 & 1 & 0  \\ \hline
 any & 0 & 0 & 0  \\ \hline
\end{tabular}
\end{table}

{\thm\label{th:55} The code $C(M)$ is self-orthogonal if and only if either $x,y\in\{0,c\}$, or $y\in\{a,b\}$ and the three parameters $\lambda,\mu,\nu$ are as in Table $\ref{tab:11}$.}

\begin{proof}
The code $ C(M)$ is self-orthogonal if and only if $GG^\mathrm{T}=0.$

If $y\in \{0,c\}$, $GG^\mathrm{T}=0$ implies $x\in\{0,c\}$.
It is trivial because the code only has a zero codeword.

If $y\in\{a,b\}$,
then
$$\begin{aligned}
GG^\mathrm{T}&=x^2I+y^2MM^\mathrm{T}\\
&=x^2I+y(\lambda I+\mu J+\nu M)\\
&=x^2I+y\lambda I+y\mu J+y\nu M.
\end{aligned}
$$
Therefore, $GG^\mathrm{T}=0$ if and only if $-y\nu M=x^2I+y\lambda I+y\mu J.$
\\
Since $$
\begin{aligned}
(-y\nu M)(-y\nu M)^\mathrm{T}&=y\nu MM^\mathrm{T}\\
[(x^2+y\lambda)I+y\mu J][(x^2+y\lambda)I+y\mu J]^\mathrm{T}&=y\nu(\lambda I+\mu J+\nu M)\\
(x^2+y\lambda)I+ny\mu J&=y\nu \lambda I+y\nu\mu J+y\nu M\\
&=y\nu\lambda I+y\nu\mu J-(x^2+y\lambda)I-y\mu J,
\end{aligned}
$$
then we have $(n+1-\nu)\mu yJ=y\nu\lambda I.$ Because $J$ is the all-one matrix, and $I$ is the identity matrix,
then
\begin{equation}
\left\{
\begin{array}{l}
y\nu\lambda=0\\
(n+1-\nu)\mu y=0.
\end{array}
\right.
\end{equation}
Thus, $\lambda,\mu,\nu \in \mathbb{F}_2$ are as in Table \ref{tab:11}.
\end{proof}

The next two results give conditions for $C(M)$ to be QSD (resp. Type IV).
{\thm\label{th:6} A self-orthogonal code $C(M)$ is QSD if and only if either $x\in\{a,b\}$, or $x\in\{0,c\}, y\in\{a,b\}, \lambda=\mu=\nu=0$ and $M$ is a full rank matrix spanning a self-orthogonal binary code.}

\begin{proof}
A self-orthogonal code $C$ is QSD if and only if $G$ has $n$ linearly independent rows.
If $x\in\{a,b\}$, $C$ is QSD because of the form of $G$.

If $x\in\{0,c\}$, then we must let the determinant $|yM|\neq0$ to make sure there are $n$ linearly independent rows in $G$.
From the proof of Theorem \ref{th:55}, $y\in\{a,b\}$ and $-y\nu M=y\lambda I+y\mu J$.
Then, $\lambda=\mu=\nu=0$, and $M$ is a binary matrix such that
\begin{equation}\label{eq:2}
\left\{
\begin{array}{l}
|M|\neq0,\\
MM^\mathrm{T}=0.
\end{array}
\right.
\end{equation}

This completes the proof.
\end{proof}

{\thm \label{th:7}A QSD code $C(M)$ is Type IV if either $x\in\{0,c\},$ or $x\in\{a,b\}$ and one of the following three conditions holds.
\begin{enumerate}
\renewcommand{\labelenumi}{$(\theenumi)$}
 \item $\lambda=\mu=0, \nu=1,$
 \item $\lambda=0, \mu=1, \nu=1,$
 \item $\lambda=1,\mu=\nu=0.$
\end{enumerate}

}
\begin{proof}
It easy to check that a QSD code is Type IV if the generator matrix $G$ has all the rows of even weights.
If $x\in\{0,c\}$, then $\lambda=\mu=\nu=0$ because of Theorem \ref{th:6}. From $MM^\mathrm{T}=0$ in Equation \ref{eq:2}, it is clear that $M$ has all the rows of even weights.

If $x\in\{a,b\}$, we just prove that $M$ has all the rows of odd weights in the three cases. Now we have
\begin{equation}
\left\{
\begin{array}{l}
MM^\mathrm{T}=\lambda I+\mu J+\nu M,\\
y\nu M=(x+y\lambda)I+y\mu J.
\end{array}
\right.
\end{equation}
\begin{enumerate}
\renewcommand{\labelenumi}{(\theenumi)}
 \item  $\lambda=\mu=0, \nu=1$.
In this case, we have
$\left\{
\begin{array}{l}
MM^\mathrm{T}=M,\\
yM=xI.
\end{array}
\right.$
Therefore, $x=y$ and $M=I$ with all rows of odd weights.
 \item  $\lambda=0, \mu=1, \nu=1$.
In this case, we have $\left\{
\begin{array}{l}
MM^\mathrm{T}=M+J,\\
yM=xI+yJ.
\end{array}
\right.$
Therefore, $x=y$ and $M=J-I$ with even $n$. So, $M$ has all rows of odd weights.
\item $\lambda=1,\mu=\nu=0$.
In this case, we have $\left\{
\begin{array}{l}
MM^\mathrm{T}=I\\
(x+y)I=0
\end{array}
\right.$.
Therefore, $x=y$ and $M$ has all rows of odd weights.
\end{enumerate}

This completes the proof.
\end{proof}

We now investigate the residue and torsion codes of $C(M)$.

From \cite[Thm. 1]{E}, we write the generator matrix in the form
$$G=\left(\begin{array}{ccc}
aI_{k_1}&X&Y\\
0&cI_{k_2}&cZ
\end{array} \right).$$
For $x\neq0$, we have the following cases depending on the values of $x \in E.$
\begin{itemize}
 \item  If $x=a$ or $x=b$, then $k_1=n,k_2=0,(X,Y)=yM.$ The generator matrix of the residue code is $(I,M)$ if $y=a,b$ and $(I,\mathbf{0})$ if $y=c.$

 \item If $x=c$, then, $k_1=0,k_2=n,y=c,Z=M.$  The generator matrix of the torsion code is $G_2=(I,M)$.

\end{itemize}

The (additive) generator matrix of the corresponding additive $\F_4$ code is
$$G'=\left(\begin{array}{c}
 \phi(aG)\\
\phi(bG)
\end{array} \right),$$
where  $\phi$ is as defined in the preceding section.

{\bf Remarks:}
\begin{itemize}
 \item If $y=c$, then $C(M)$ has minimum distance $1.$ In the examples, we shall assume that $y=a,$ or $y=b.$

 \item If $x=c$, then we find that $\phi(M)$  is a linear code over $\F_4$ given by $\phi(M) =\langle (0,M) \rangle.$ We will avoid this case as well.

 \item Now if both $x,y$ are in $\{a,b\},$ then we find that $\phi(M)$ is a linear code over $\F_4$ given by $\phi(M) =\langle (I,M) \rangle.$
\end{itemize}

\section{Combinatorial matrices}

\subsection{Two-class association schemes}

From now on, we can discuss two-class association schemes which will play an important role in $M$.

There are two kinds of two-class association schemes. One is a Strongly Regular Graph (SRG), where the two adjacency matrices satisfy $A_i = A_i^T$ for $i=1,2$.
Here, $A_2$ satisfies $A_2 = J-I-A_1:=\overline{A_1}$. An important database of SRGs is \cite{P}.

A classical construction of a SRG is the Paley graph. It is constructed from quadratic residues in $\mathbb F_q$, where $q \equiv 1 \pmod 4$ and $A=Q=N$.
The parameters are $(q, \frac{q-1}{2}, \frac{q-3}{4}, \frac{q+1}{4})$. The example of $q=5$ is the pentagon graph.

Another class is a Doubly Regular Tournament (DRT), which is equivalent to a skew Hadamard matrix \cite{B}. The adjacency matrix $A_2$ satisfies $A_2 = J-I-A_1:=\overline{A_1}$. Note that $A_1^T = \overline{A_1}$.

From now on, let $A=A_1$.

\begin{lem}[\cite{T}]\label{lem:1}
If $G$ is an SRG, then we have
\[
AA^T = A^2 = \kappa I + \Lambda A + \Mu \overline{A}.
\]

If $G$ is a DRT, then we have
\[
AA^T = \kappa I + (\kappa-1-\Lambda) A + (\kappa-\Mu) \overline{A}.
\]
\end{lem}

Using the same parameters in the above lemma, both of them satisfy the equation
$$AJ=JA=\kappa J,$$
and for SRGs, we have
\begin{equation}\label{eq:4}
A^2=\kappa I+\Lambda A+\Mu (J-I-A),
\end{equation}
for DRTs, we have
\begin{equation}\label{eq:5}
A^2=\Lambda A+\Mu(J-I-A).
\end{equation}

We connect these parameters to that of the matrix $M$ of the preceding section. The trivial proof is omitted.

{\prop Keep the notation of Lemma $\ref{lem:1}$. If $M$ is the adjacency matrix of $G$ with parameters $(n,\kappa,\Lambda,\Mu)$ then

\begin{itemize}
 \item in the SRG case $\lambda=\kappa-\Mu, \mu=\Mu,\, \nu= \Lambda-\Mu$,
 \item in the DRT case $\lambda=\Mu, \mu=\kappa-\Mu,\, \nu=\Mu- \Lambda-1$.
\end{itemize}

}

We can use the database of two class association schemes from Hanaki and Miyamoto's database~\cite{HM}. In particular there is a classification of DRT of sizes up to 40.

\subsection{Pure and double circulant codes from two-class association schemes}

We can also follow the construction method from~\cite{T}. Let $Q_E(r,s,t)=rI + sA + t\overline{A}$, where $r,s,t \in E$, where $A$ is an adjacency matrix of a SRG or a DRT.
 Let $C(Q_E(r,s,t))$ be a code of length $2n$ with a generator matrix of the form

$$G=(aI, Q_E(r, s, t))=(aI,rI + sA + t\overline{A}).$$

This construction can be called the {\it pure} construction.

\medskip

First we consider $r=0$ and $s,t \in \{a, b \}$.
The code $C(Q_E(0, s, t))$ of length $2n$ has  generator matrix of the form
$$G=(aI, Q_E(0, s, t))=(aI,sA + t\overline{A}),$$
where $A$ is an adjacency matrix of a SRG or a DRT.

{\thm Suppose $A$ is an adjacency matrix of a SRG or a DRT.

\begin{enumerate}
\renewcommand{\labelenumi}{$(\theenumi)$}
\item If $n \ge 7$, then the minimum distance of $C(Q_E(0, s, t))$ is exactly $4$.

    \item
    If $3 \le n < 7$, then the minimum distance of $C(Q_E(0, s, t))$ is $2$ or $3$.
\end{enumerate}
}

\begin{proof}
Due to symmetry between $A$ and $\overline{A},$ we may assume $s=a$ and $t=b$, or $s=t=a$. We only consider the case $s=a$ and $t=b$ because the other case $s=t=a$ can be done similarly.
Note that $aG=a(aI,aA + b\overline{A})=(aI, aA+ a\overline{A})=(aI, a(A+\overline{A}))$.
Since $A+\overline{A}=J-I$, $aG=(aI, a(J-I))$. It is easy to see that the minimum distance of the code generated by $(aI, a(J-I))$ is 4 if $n \ge 3$. Hence $G$ generates a codeword of weight 4 if $n \ge 3$. Each row of $G$ and $aG$ has weight at least 4 if $n \ge 7$. Hence the first statement of the theorem follows. If $3 \le n <7$, then $G$ has weight $2$ or $3$. Hence the second statement  follows.
\end{proof}

Similarly, we have the following theorem.

{\thm Suppose $A$ is an adjacency matrix of a SRG or a DRT. If $r \ne 0$, and $s,t \in \{a, b \}$, then the following statements hold.

\begin{enumerate}
\renewcommand{\labelenumi}{$(\theenumi)$}
\item If $n \ge 7$ and $r=c$, then the minimum distance of $C(Q_E(r, s, t))$ is exactly $4$.

    \item
    If $r$ is either $a$ or $b$, then the minimum distance of $C(Q_E(r, s, t))$ is $2$.
\end{enumerate}
}

Therefore if $n \ge 7$, it is reasonable to consider the following three constructions
(i) $C(Q_E(0, a, 0))$, (ii) $C(Q_E(a, a, 0))$, or
(iii) $C(Q_E(c, a, 0))$, where replacing $a$ into $b$ gives the same result.

Note that Case (i) and Case (ii) are the same construction as $C(M)$ with $x=a$ and $y=a$ in Section \ref{sec:3} by taking $M=A$ and $M=A+I$, respectively. Therefore, we can apply these two cases to various SRGs and DRTs.

\medskip

Next we can consider the {\it bordered} construction as follows.

$$
B_E(r,s,t)=
\left(
\begin{array}{c|ccc|c|c}
a & 0& \dots & 0 & 0 & a ~ \dots ~ a \\
\hline
0 &  &       &   & a &     \\
\vdots &     &  aI &  & \vdots  & Q_E(r,s,t)  \\
0 &  &       &   & a &            \\
\end{array}
\right).
$$
Just like for the pure construction, we can distinguish three cases
(i) $Q_E(0, a, 0)$, (ii) $Q_E(a, a, 0)$, or
(iii) $Q_E(c, a, 0)$.

\begin{lem}
The codes in these two constructions with Case $({\rm i})$ and Case $({\rm iii})$ are the same.
\end{lem}
\begin{proof}
In Case $({\rm i})$, $Q_E(0,a,0)=aA$, and the generator matrix $G_{({\rm i})}=(aI|aA)$ in pure construction.
So, the code $C_{({\rm i})}=\{\mathbf{x}G_{({\rm i})}|\mathbf{x}\in E^{n}\}$.
In Case $({\rm iii})$, $Q_E(c,a,0)=cI+aA$, and the generator matrix $G_{({\rm iii})}=(aI|cI+aA)$ in pure construction.
So, the code $C_{({\rm iii})}=\{\mathbf{x}G_{({\rm iii})}|\mathbf{x}\in E^{n}\}$.
Since
$$\begin{aligned}
\mathbf{x}G_{({\rm iii})}&=\mathbf{x}(G_{({\rm i})}+(\mathbf{0}|cI))\\
&=\mathbf{x}G_{({\rm i})}+\mathbf{x}(\mathbf{0}|cI)\\
&=\mathbf{x}G_{({\rm i})},
\end{aligned}
$$
we have $C_{({\rm i})}=C_{({\rm iii})}$.

For bordered construction in Case (i), we have
$C^\prime_{({\rm i})}=\{\mathbf{y}G^\prime_{({\rm i})}|\mathbf{y}\in E^{n+1}\},$
where
$$G^\prime_{({\rm i})}=\left(
\begin{array}{c|ccc|c|c}
a & 0& \dots & 0 & 0 & a ~ \dots ~ a \\
\hline
0 &  &       &   & a &     \\
\vdots &     &  aI &  & \vdots  & aA  \\
0 &  &       &   & a &            \\
\end{array}
\right),$$
and in Case (iii), we have $C^\prime_{({\rm iii})}=\{\mathbf{y}G^\prime_{({\rm iii})}|\mathbf{y}\in E^{n+1}\},$
where
$$G^\prime_{({\rm iii})}=\left(
\begin{array}{c|ccc|c|c}
a & 0& \dots & 0 & 0 & a ~ \dots ~ a \\
\hline
0 &  &       &   & a &     \\
\vdots &     &  aI &  & \vdots  & cI+aA  \\
0 &  &       &   & a &            \\
\end{array}
\right).$$
Let
$$A^\prime=\left(
\begin{array}{c|ccc|c|c}
0 & 0& \dots & 0 & 0 & 0 ~ \dots ~ 0 \\
\hline
0 &  &       &   & 0 &     \\
\vdots &     &  \mathbf{0} &  & \vdots  & cI  \\
0 &  &       &   & 0 &            \\
\end{array}
\right),$$
then
$$\begin{aligned}
\mathbf{y}G^\prime_{({\rm iii})}&=\mathbf{y}(G^\prime_{({\rm i})}+A^\prime)\\
&=\mathbf{y}G^\prime_{({\rm i})}+\mathbf{y}A^\prime\\
&=\mathbf{y}G^\prime_{({\rm i})}.
\end{aligned}
$$
Therefore, $C^\prime_{({\rm i})}=C^\prime_{({\rm iii})}$.
\end{proof}

\begin{examp}
{\rm
It is well known that there is unique DRT of order $11$. The pure construction with $Q_E(0, a, 0)$ gives a QSD $[22, 11, 6]$ code over $E$. The bordered construction with $Q_E(a, a, 0)$ gives a QSD $[24, 12, 8]$ code over $E$. The minimum distances of these codes are justified by Theorem~$\ref{thm-min-dinstance}$.
}
\end{examp}

\begin{lem}\label{l:22}
\begin{enumerate}
\renewcommand{\labelenumi}{$(\theenumi)$}
\item For SRGs we have $$
Q_E(r,s,t)Q_E(r,s,t)^T=\omega_1I+\omega_2A+\omega_3\overline{A},
$$
where
$\omega_1=(r^2+s^2\kappa-t^2-t^2\kappa+t^2v)$,
$\omega_2=(rs+sr+s^2\Lambda-st-ts-st\Lambda-ts\Lambda+t^2\Lambda+st\kappa+ts\kappa+t^2v-2t^2\kappa)$,
$\omega_3=(rt+tr+s^2\Mu-st\Mu-ts\Mu+t^2\Mu+st\kappa+ts\kappa+t^2v)$.

\item For DRTs we have $$
Q_E(r,s,t)Q_E(r,s,t)^T=\omega_1^\prime I+\omega_2^\prime A+\omega_3^\prime \overline{A},
$$
where
$\omega_1^\prime=(r^2+(s^2+t^2)\kappa)$,
$\omega_2^\prime=(rt+sr+s^2(\kappa-1-\Lambda)+t^2(\kappa-\Mu)+st\Lambda+st\Mu)$,
$\omega_3^\prime=(tr+rs+s^2(\kappa-\Mu)+t^2(\kappa-1-\Lambda)+st\Mu+st\Lambda)$.
\end{enumerate}
\end{lem}
\begin{proof}
It is straightforward by Equations \ref{eq:4} and \ref{eq:5} and Lemma \ref{lem:1}.
\end{proof}

We will discuss the weight of rows of generator matrices in Case (i) and Case (ii).
Then, the conditions of QSD and Type IV can be confirmed.
By the form of generator matrices in pure construction and bordered construction, the code is QSD if it is self-orthogonal.
The following remark gives when the code is self-orthogonal and Type IV.
\begin{rem}
{\rm
For Cases (i) and (ii), we have the following observations.

\begin{itemize}
\item \underline{pure construction with SRGs}

For the code $P_E(r,s,t)$ to be self-orthogonal, we need $$(aI| Q_E(r, s, t))(aI | Q_E(r, s, t))^T = \mathbf{0}.$$
That is we need $Q_E(r, s, t)Q_E(r, s, t)^T = -aI$.
By Lemma $\ref{l:22}$ (1), we compute the parameters $\kappa,\Lambda,\Mu$ of self-orthogonal (QSD) codes in Table $\ref{t:2}$.

The weight of any row of $Q_E(r,s,t)$ is related to the coefficient of $I$, where $I$ is in Lemma $\ref{l:22}$ (1).
So, the weight of any row of $(aI|Q_E(r,s,t))$ is $$1+\alpha(r^2)+\alpha(s^2)\kappa+\alpha(t^2)(n-\kappa-1),$$
that is, $1+\kappa$ in Case (i) and $2+\kappa$ in Case (ii).
Therefore, a QSD code is Type IV if $$1+\alpha(r^2)+\alpha(s^2)\kappa+\alpha(t^2)(n-\kappa-1)=0\ ({\rm mod}\ 2),$$
that is, $1+\kappa=0\ ({\rm mod}\ 2)$ in Case (i) and $2+\kappa=0\ ({\rm mod}\ 2)$ in Case (ii).
Then we have the conditions of Type IV in Table $\ref{t:2}$.

\item \underline{bordered construction with SRGs}

Similar to the pure construction, we need
$$
B_E(r,s,t)B_E(r,s,t)^T=\mathbf{0}.
$$
Then we have
$$
\begin{aligned}
a(1+n)&=0\\
a(r+s\kappa+t(n-\kappa-1))&=0\\
aI+aJ+Q_E(r,s,t)Q_E(r,s,t)^T&=\mathbf{0}.
\end{aligned}
$$
The first equation is the product of the top row with itself.
The second equation is the product of the top row with any other row,
and the third equation ensures that the other rows are orthogonal to each other.
The results of the calculation by Lemma $\ref{l:22}$ (1) are in Table $\ref{t:2}$.

And this code is Type IV if $$
\begin{aligned}
\alpha(a)(1+n)&=0\ ({\rm mod}\ 2),\\
\alpha(r)+\alpha(s)\kappa+\alpha(t)(n-\kappa-1)&=0\ ({\rm mod}\ 2).
\end{aligned}
$$
We also have the results in Table $\ref{t:2}$.
\begin{table}[!htbp]
\caption{Conditions of QSD and Type IV with SRGs}\label{t:2}
\begin{tabular}{|c|c|c|c|c|c|c|}
\hline
\multirow{2}{*}{$r$} & \multirow{2}{*}{$s$} & \multirow{2}{*}{$t$} & \multicolumn{2}{c|}{Pure}          & \multicolumn{2}{c|}{Bordered}           \\ \cline{4-7}
                     &                      &                      & QSD                      & Type IV & QSD                           & Type IV \\ \hline
0                    & $a$                  & 0                    & $\kappa=1,\Lambda=\Mu=0$ & Always  & $\kappa=0,n=\Lambda=\Mu=1$ & Always  \\ \hline
$a$                  & $a$                  & 0                    & $\kappa=\Lambda=\Mu=0$   & Always  & $n=\Lambda=\Mu=\kappa=1$      & Always  \\ \hline
\end{tabular}
\end{table}
\item \underline{pure and bordered construction with DRTs}

By using the same arguments as these two constructions with SRGs and Lemma $\ref{l:22}$ (2),
then we have the results in Table \ref{t:3}.

\begin{table}[!htbp]
\caption{Conditions of QSD and Type IV with DRTs}\label{t:3}
\begin{tabular}{|c|c|c|c|c|c|c|}
\hline
\multirow{2}{*}{$r$} & \multirow{2}{*}{$s$} & \multirow{2}{*}{$t$} & \multicolumn{2}{c|}{Pure}          & \multicolumn{2}{c|}{Bordered}        \\ \cline{4-7}
                     &                      &                      & QSD                      & Type IV & QSD                        & Type IV \\ \hline
0                    & $a$                  & 0                    & $\kappa=\Mu=1,\Lambda=0$ & Always  & $\kappa=\Lambda=0,n=\Mu=1$ & Always  \\ \hline
$a$                  & $a$                  & 0                    & $\kappa=\Lambda=0,\Mu=1$ & Always  & $n=\Mu=\kappa=1,\Lambda=0$ & Always  \\ \hline
\end{tabular}
\end{table}
\end{itemize}
}
\end{rem}

We computed the Hamming weight and Lee weight of some codes.
These examples are from \cite{HM,P} and MAGMA databases of SRGs \cite{DS}.

\begin{thm}
There are  QSD codes over $E$ with the following parameters.
\begin{enumerate}
\renewcommand{\labelenumi}{$(\theenumi)$}
\item Based on SRGs, there are QSD codes with parameters $(2n, d)$, where $2n$ is the length of the code, and $d$ is the minimum distance.
    $$(32,8),(56,8),(70,10),(72,12),(80,12),(92,12).$$

\item Based on DRTs, there are QSD codes with parameters $(2n, d)$, where $2n$ is the length of the code, and $d$ is the minimum distance.
    $$(22,7),(24,8),(38,8),(40,8).$$
\end{enumerate}
\end{thm}

We display these results in Table \ref{ta:4} and Table \ref{ta:5}.

\begin{table}[!htpb]
\caption{Weights of some QSD codes of SRGs}\label{ta:4}
\begin{tabular}{|c|c|c|c|c|c|}
\hline
\multicolumn{1}{|l|}{Construction} & Cases                 & $(n-\kappa-\Lambda-\Mu)$ & Code Length & Hamming  & Lee  \\ \hline
\multirow{6}{*}{Pure}              & (i)                   & $(36-15-6-6)$            & 72          & 12             & 12         \\ \cline{2-6}
                                   & \multirow{5}{*}{(ii)} & $(16-6-2-2)$             & 32          & 8              & 8          \\ \cline{3-6}
                                   &                       & $(28-12-6-4)$            & 56          & 8              & 8          \\ \cline{3-6}
                                   &                       & $(35-16-6-8)$            & 70          & 10             & 10         \\ \cline{3-6}
                                   &                       & $(36-14-4-6)$            & 72          & 12             & 12         \\ \cline{3-6}
                                   &                       & $(40-12-2-4)$            & 80          & 12             & 12         \\ \hline
\multirow{3}{*}{Bordered}          & \multirow{3}{*}{(i)}  & $(15-6-1-3)$             & 32          & 8              & 8          \\ \cline{3-6}
                                   &                       & $(27-10-1-5)$            & 56          & 8              & 8          \\ \cline{3-6}

                                   &                       & $(45-12-3-3)$            & 92          & 12             & 12         \\ \hline
\end{tabular}
\end{table}

\begin{table}[!htpb]
\caption{Weights of QSD some codes of DRTs}\label{ta:5}
\begin{tabular}{|c|c|c|c|c|c|c|c|c|}
\hline
Construction              & $n$  &Length & \multirow{5}{*}{\begin{tabular}[c]{@{}c@{}}Case\\ (i)\end{tabular}} & Hamming & Lee & \multirow{5}{*}{\begin{tabular}[c]{@{}c@{}}Case\\ (ii)\end{tabular}} & Hamming & Lee \\ \cline{1-3} \cline{5-6} \cline{8-9}
\multirow{2}{*}{Pure}     & $11$ & $22$        &                                                                     & 6       & 6   &                                                                      & 7       & 7   \\ \cline{2-3} \cline{5-6} \cline{8-9}
                          & $19$ & $38$        &                                                                     & 8       & 8   &                                                                      & 7       & 7   \\ \cline{1-3} \cline{5-6} \cline{8-9}
\multirow{2}{*}{Bordered} & $11$ & $24$        &                                                                     & 7       & 7   &                                                                      & 8       & 8   \\ \cline{2-3} \cline{5-6} \cline{8-9}
                          & $19$ & $40$        &                                                                     & 8       & 8   &                                                                      & 8       & 8   \\ \hline
\end{tabular}
\end{table}

The  images by $\phi()$ of these codes are formally self-dual additive codes over $\F_4$  in \cite{HK}.
\section{Conclusion}
In this work, we have constructed QSD and Type IV codes over the ring $E$ in the sense of \cite{E}. The construction method is based on the adjacency matrices of two-class association schemes,
in an analogue over $E$  of \cite{T} over finite fields. Formally self-dual additive codes over $\F_4$ were introduced in \cite{HK}. This little-known class of codes deserves further exploration.
In another direction, the construction methods we used can be explored over  the rings $H$ and $I$ of the Raghavendran classification \cite{F}.

\end{document}